%% file: main.tex
\documentclass[letterpaper, 10 pt, conference]{ieeeconf}

\usepackage{amsmath}

\usepackage{amsthm}
\usepackage{amsfonts}
\usepackage{mathtools}
\usepackage{booktabs}
\usepackage{algorithm}
\usepackage{algpseudocode}
\usepackage{float}
\usepackage{xcolor}
\usepackage{graphicx}
\graphicspath{{figures/}}
\usepackage{import}
\usepackage{layouts}
\usepackage{cite}

\newcommand{\R}{\mathbb{R}}

\DeclareMathOperator{\rank}{rank}
\DeclareMathOperator{\diag}{diag}

\newtheorem{assumption}{Assumption}
\newtheorem{theorem}{Theorem}
\newtheorem{definition}{Definition}
\newtheorem{remark}{Remark}
\newtheorem{lemma}{Lemma}

\newtheorem{problem}{Problem}

\newif\ifcomments
\commentstrue 

\IEEEoverridecommandlockouts                  

\title{\LARGE \bf
Robust Data-Driven Control for Nonlinear Systems\\Using their Digital Twins and Quadratic Funnels\thanks{This work has been supported by NSF Award EPCN 2149470 and the 2024-1.2.3-HU-RIZONT-2024-00030 project.}
}

\author{Shiva Shakeri and Mehran Mesbahi$^{1}$
\thanks{$^{1}$The authors are with the William E. Boeing Department of Aeronautics and Astronautics at the University of Washington; Emails: {\tt \{sshakeri+mesbahi\}@uw.edu}.}
}

\begin{document}

\maketitle
\thispagestyle{empty}
\pagestyle{empty}

\begin{abstract}
This paper examines a robust data-driven approach for the safe deployment of systems with nonlinear dynamics using their imperfect digital twins. Our contribution involves proposing a method that fuses the digital twin's nominal trajectory with online, data-driven uncertainty quantification to synthesize robust tracking controllers. Specifically, we derive data-driven bounds to capture the deviations of the actual system from its prescribed nominal trajectory informed via its digital twin. Subsequently, the dataset is used in the synthesis of quadratic funnels—robust positive invariant tubes around the nominal trajectory—via linear matrix inequalities built on the time-series data. The resulting controller guarantees constraint satisfaction while adapting to the true system behavior through a segmented learning strategy, where each segment's controller is synthesized using uncertainty information from the previous segment. This work establishes a systematic framework for obtaining safety certificates in learning-based control of nonlinear systems with imperfect models.
\end{abstract}

\section{Introduction}\label{sec:intro}
In recent years, direct data-driven control has attracted significant interest, particularly in the context of using finite trajectory data to make online control decisions while bypassing explicit model identification \cite{Hou2013FromMC}. A central question is how to \emph{represent unknown dynamics from (offline and online) data} so that prediction and decision-making can be posed directly on measured trajectories. In the linear time-invariant (LTI) setting, Willems’ Fundamental Lemma \cite{Willems2005FundamentalLemma} shows that, under a persistently exciting experiment, one can build a Hankel matrix whose column space spans all trajectories of a given length, enabling data-driven prediction and constraint handling (e.g., DeePC \cite{Coulson2019DeePC} and data-driven tracking MPC \cite{Berberich2022LTMPCCD}). For noisy data, a matrix S-lemma reduces data-consistency and performance quadratic inequalities to non-conservative linear matrix inequalities \cite{Boyd1994LMI}, enabling feedback synthesis directly from noisy input/state trajectories with certificates over all models consistent with the data \cite{vanWaarde2020MatrixSLemma}. Beyond LTI, recent work develops direct data-driven methods for linear time-varying (LTV) dynamics: for finite-horizon optimal control, \cite{Pang2018FiniteOptimal} proposes off-policy and online policy-iteration schemes for unknown discrete-time LTV systems; for state-feedback from trajectories, \cite{Nortmann2020CDC_LTV} derives convex, data-dependent design conditions, later extended in \cite{Nortmann2023TAC_LTV} to handle simultaneous process/measurement noise and periodic LTV cases.

LTV models arise naturally in nonlinear systems under changing operating conditions, and via linearizing along a nominal trajectory at time-varying operating points \cite{Kamen2010LTV}. This locally LTV deviation viewpoint (e.g., \cite{Reynolds2020TIFS}) underpins set-based safety methods for nonlinear tracking—tube MPC with bounded disturbances \cite{MAYNE2005219,RAKOVIC20121631} and SOS-based funnel certificates grounded in semidefinite relaxations \cite{Parrilo2003,Reynolds2020}—which synthesize time-varying sets around a nominal path and certify invariance. While powerful, these approaches typically presume a trusted dynamics model and substantial offline computation, making it challenging to adapt from finite, online trajectory data.

Digital twins enable nominal planning and policy prototyping prior to deployment \cite{Asch2022-iw,Zhang2024-wc}. Yet a twin is inevitably imperfect, so safely transferring nominal plans to the physical plant requires online quantification and mitigation of model–plant mismatch. This need is particularly acute when plant access is limited, risky, or costly; see also \cite{Wickramasinghe2025-us}.

\textbf{Contributions.} We develop a direct, online, certificate-based data-driven framework for constrained trajectory tracking of nonlinear plants from finite trajectories, without an identification step. Assuming access to a nominal trajectory computed on a known digital twin, we model plant–twin deviation as a locally LTV system and, from measured deviations, construct data-consistent uncertainty sets that capture linearization error and variation of the LTV matrices. Over these sets we synthesize time-varying ellipsoidal invariant tubes (funnels) via a matrix S-lemma reduction to LMIs, yielding non-conservative certificates that every closed-loop evolution remains within a robust positively invariant tube while respecting state and input constraints. The scheme operates online over successive time segments: within each segment a fixed feedback is applied while deviation data are gathered; at each boundary, the newly accrued data update the uncertainty set, and a single convex semidefinite program, minimizing tube volume subject to state and input limits, returns the tube and feedback for the next segment. We also provide informativity conditions specifying when the available data are sufficient for certification.

The remainder of this paper is organized as follows.
Section~\ref{sec:problem_statement} states the problem and assumptions.
Section~\ref{sec:preliminaries} develops the deviation-based modeling and basic bounds.
Section~\ref{sec:data_driven_synthesis} introduces the online, segmented learning setup and the resulting data-consistent uncertainty description.
Section~\ref{sec:data_driven_funnel} presents the funnel-synthesis framework, the associated optimization, and the main algorithmic and stability results.
Section~\ref{sec:case} presents the case study, and Section~\ref{sec:conclusion} concludes.

\paragraph*{Notation}
$M^\top$ denotes transpose and $M^{-1}$ the inverse (when nonsingular). $I_n$ is the $n\times n$ identity and $0_{n\times m}$ the $n\times m$ zero (sizes omitted when clear). For symmetric $A,B$, we use the Loewner order: $A\succeq0$ ($\succ0$) means positive (semi)definite and $A\preceq B\iff B-A\succeq0$. Extremal eigenvalues are $\lambda_{\min}(A),\lambda_{\max}(A)$. Norms: for vectors, $\|x\|$ is Euclidean; for matrices, $\|M\|$ is the induced spectral norm ($\|M\|_2$) unless stated, and $\|M\|_F$ is Frobenius. For a sequence $z(k)$ on $\mathcal S$, $\|z\|_{\infty,\mathcal S}:=\max_{k\in\mathcal S}\|z(k)\|$. For a function $g:\mathcal S\to\mathbb R^p$, $\|g\|_{\infty,\mathcal S}:=\sup_{s\in\mathcal S}\|g(s)\|$ (inner norm Euclidean). $\rank(\cdot)$ is rank; $\diag(\cdot)$ forms (block-)diagonals. $\oplus$ is the Minkowski sum; $\log\det(\cdot)$ is the log-determinant. Discrete time is indexed by $k\in\mathbb Z_{\ge0}$.

\section{Problem Statement}
\label{sec:problem_statement}

Consider an unknown discrete-time nonlinear system, the \emph{physical plant}, with the unknown dynamics:
\begin{equation}\label{eq:phys_dt}
    x(k+1) = f\big(x(k), u(k)\big),
\end{equation}
where \(x(k) \in \mathcal{X} \subset \mathbb{R}^n\) is the state and \(u(k) \in \mathcal{U} \subset \mathbb{R}^m\) is the control input. The sets \(\mathcal{X}\) and \(\mathcal{U}\) represent known, compact constraint sets on the state and input, respectively.

We assume that we have access to a \emph{digital twin} of the plant, with dynamics given by,
\begin{equation}\label{eq:twin_dt}
\hat{x}({k+1})= \hat{f}\big(\hat{x}(k), \hat{u}(k)\big),
\end{equation}
where \(\hat{x}(k) \in \mathcal{X}\) and \(\hat{u}(k) \in \mathcal{U}\) are the twin's state and input. The interconnection and data flow between the physical plant and its digital twin are illustrated in Fig.~\ref{fig:digital-twin}.

We assume that the plant dynamics \(f\) and the model \(\hat{f}\) are related through an additive mismatch \(\Delta: \mathcal{X} \times \mathcal{U} \to \mathbb{R}^n\) such that,
\begin{equation}\label{eq:mismatch}
f(x, u) = \hat{f}(x, u) + \Delta(x, u), \quad \forall (x, u) \in \mathcal{X} \times \mathcal{U}.
\end{equation}
The function \(\Delta\) encapsulates the discrepancies between the plant and its twin, including unmodeled dynamics, parametric uncertainty, and external disturbances.

\begin{figure}[H]
  \centering
  \includegraphics[width=0.65\columnwidth]{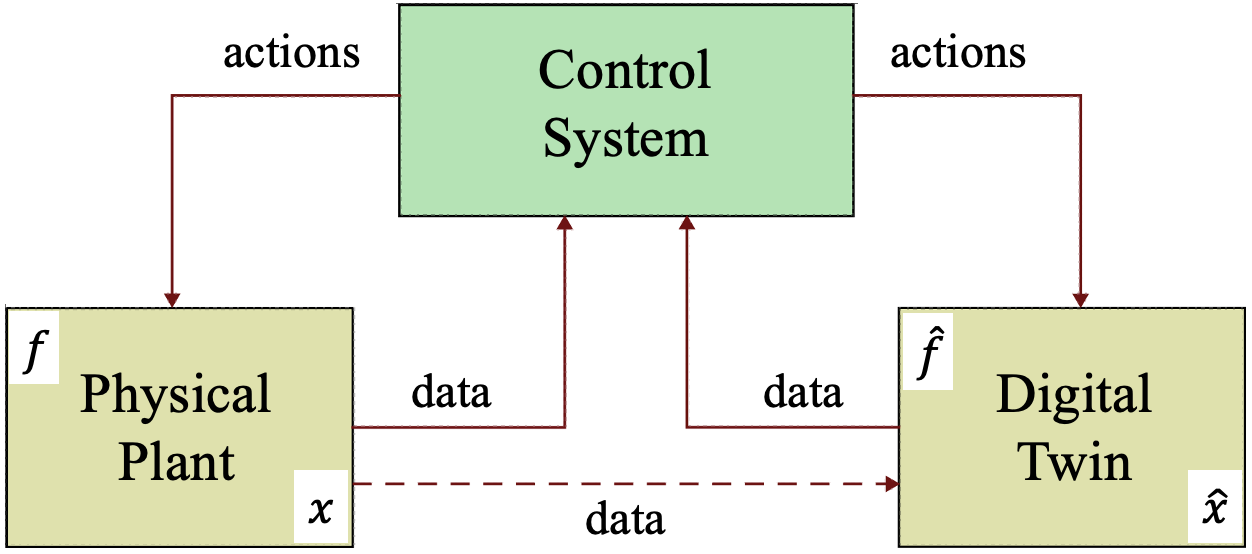}
  \caption{Interconnection and data flow between the physical plant and its digital twin.}
  \label{fig:digital-twin}
\end{figure}

\subsection{Modeling Assumptions}
\label{subsec:modeling_assumptions}

In order to derive bounds on the performance of our data-driven control scheme, we will build on the following assumptions for our subsequent analysis.
\begin{assumption}[Uniform Mismatch Bound]\label{ass:mismatch_uniform}
    The discrepancy function $\Delta$ between the plant and its digital twin is uniformly bounded. Formally, there exists a known constant \(\gamma \ge 0\) such that
    \[
    \|\Delta(x, u)\| \le \gamma \qquad \forall (x, u) \in \mathcal{X} \times \mathcal{U}.
    \]
\end{assumption}

\begin{assumption}[Smoothness]\label{ass:smoothness}
    The dynamics of the physical plant is continuously differentiable, i.e., \(f\in C^{1}\) \cite{Rudin1976PMA}.
\end{assumption}

\begin{assumption}[Feasible Nominal Trajectory]\label{ass:nominal_traj}
    A feasible nominal state-input trajectory for the digital twin is known a priori. Specifically, there exists a horizon \(N \in \mathbb{Z}_{\ge 1}\) and sequences \(\{\hat{x}_{\text{nom}}(k)\}_{k=0}^{N} \subseteq \mathcal{X}\), \(\{\hat{u}_{\text{nom}}(k)\}_{k=0}^{N-1} \subseteq \mathcal{U}\) that satisfy the digital twin dynamics \eqref{eq:twin_dt}.
    Furthermore, for \(j = 0, \dots, N-2\), this nominal sequence has bounded increments: there exists a constant \(v \ge 0\) such that
    \[
    \|(\hat{x}_{\text{nom}}({j+1}), \hat{u}_{\text{nom}}({j+1})) - (\hat{x}_{\text{nom}}(j), \hat{u}_{\text{nom}}(j))\| \le v.
    \]
\end{assumption}
\vspace{0.5em}
\begin{assumption}[Data Availability]\label{ass:data_availability}
    One can gather an ensemble of sufficiently rich input-state data from the physical plant.
\end{assumption}

The main problem that we examine in this work can now be stated as follows:
\begin{problem}\label{prob:1}
Due to the model–plant mismatch~(\ref{eq:mismatch}) and potential initial state perturbations, the nominal input sequence ${\hat{u}_{\text{nom}}(k)}$ designed for the digital twin may fail to safely steer the physical plant~(\ref{eq:phys_dt}). We therefore seek a causal feedback policy, using state–input data collected from the plant, that
\begin{enumerate}
  \item[(a)] computes an online corrective input based on the difference between the plant’s measured state and the nominal twin’s trajectory;
  \item[(b)] guarantees that the state–input pair $(x(k),u(k))$ remains within a robust positively invariant funnel $\mathcal{F}$ around $(\hat{x}_{\text{nom}}(k),\hat{u}_{\text{nom}}(k))$ (see Def.~\ref{def:funnel});
  \item[(c)] achieves the terminal condition $x(N)=\hat{x}_{\text{nom}}(N)$.
\end{enumerate}
\end{problem}

\begin{definition}[Funnel]\label{def:funnel}
A funnel, denoted by $\mathcal{F}$, is a time-varying set in state–input space that is invariant and lies entirely inside the feasible region $\mathcal X \times \mathcal U$.
\end{definition}

Given the nominal trajectory $(\hat x_{\text{nom}}(k),\hat u_{\text{nom}}(k))$, the constraint sets $(\mathcal{X},\mathcal{U})$, and online plant data, our objective is to synthesize a segmented feedback policy and associated time-varying funnels that satisfy items (a)–(c) in Problem~\ref{prob:1}.

\section{Preliminaries}
\label{sec:preliminaries}

In order to formulate the tracking performance between the plant and its twin, we define the state and input deviations as $\eta(k) \coloneqq x(k) - \hat{x}_{\text{nom}}(k)$ and $\xi(k)  \coloneqq u(k) - \hat{u}_{\text{nom}}(k)$, respectively.

Linearizing the plant dynamics \(f\) around the nominal trajectory point \((\hat{x}_{\text{nom}}(k), \hat{u}_{\text{nom}}(k))\) and using the mismatch definition~\eqref{eq:mismatch} yields the linear time-varying (LTV) error dynamics:
\begin{equation}\label{eq:eta_ltv}
\eta(k+1) = A(k)\,\eta(k) + B(k)\,\xi(k) + d(k),
\end{equation}
where $A(k)$ and $B(k)$ are the Jacobian matrices of $f$ evaluated on \((\hat{x}_{\text{nom}}(k), \hat{u}_{\text{nom}}(k))\) and
\begin{equation}\label{eq:disturbance}
    d(k) \coloneqq r(k) + \Delta(\hat{x}_{\text{nom}}(k), \hat{u}_{\text{nom}}(k)),
\end{equation}
where \(r(k)\) is the linearization error (see Taylor’s theorem \cite{Rudin1976PMA}).

\begin{lemma}[Disturbance Bound]\label{lem:disturbance_bound}
Under Assumptions~\ref{ass:mismatch_uniform} and~\ref{ass:smoothness}, and for all $(\eta,\, \xi)$ such that $(\hat{x}_{\text{nom}}(k)+\eta,\, \hat{u}_{\text{nom}}(k)+\xi) \in \mathcal{X}\times \mathcal{U}$, the disturbance $d(k)$ in~\eqref{eq:disturbance} admits the bound:
\begin{equation}\label{eq:dk_state_dep}
\|d(k)\| \le \gamma + L_r \left\|\begin{bmatrix}\eta(k)\\ \xi(k)\end{bmatrix}\right\|^2
\end{equation}
where $L_r$ is a Lipschitz constant for the Jacobian of $f$ on the compact set $\mathcal{X}\times \mathcal{U}$; see \cite{Nesterov2004}.
\end{lemma}

\begin{lemma}[Variation in Linearization]\label{lem:AB_var}
Under Assumptions~\ref{ass:smoothness} and~\ref{ass:nominal_traj}, there exists a constant \(L_J > 0\) (a Lipschitz constant for the Jacobian of \(f\) on $\mathcal{X}\times \mathcal{U}$) such that, for all \(k, s \in \{0, \dots, N-1\}\), (see, e.g., \cite{NesterovPolyak2006}),
\[
\left\| \begin{bmatrix} A(k)-A(s) & B(k)-B(s) \end{bmatrix} \right\|_2 \le L_J \, v \, |k-s|.
\]
\end{lemma}

\section{Data-Driven Deviation Model}
\label{sec:data_driven_synthesis}

This section specifies how trajectory data are collected on successive time segments and how these data induce a deviation-based, data-consistent system description.

The proposed algorithm proceeds iteratively. The time horizon is partitioned into segments. During each segment $i$, a fixed control law is applied while the input-state data is gathered. At the end of the segment, this dataset is used to synthesize a new controller, that is applied over the {\em subsequent time} segment $i+1$. This process is then repeated, leveraging the nominal trajectory and the measured deviations, iteratively improving the control performance and its robustness. The following subsections formalize each component of the proposed approach.

\subsection{Data Collection}
\label{subsec:data_collection}

In order to enable the data-driven synthesis of the proposed robust controllers, we partition the time horizon into fixed-length segments of $T$ steps (cf.~\cite{Liu_2023}). Each segment $i$ corresponds to a control application interval where a fixed feedback gain $K_i$ is employed. The segment boundaries are defined as:
\begin{equation}
    \mathcal{T}_i = \{k_i, k_i+1, \dots, k_{i+1}-1\}, \quad k_i = iT, \quad i \in \mathbb{Z}_{\ge 0}.
\end{equation}
During $\mathcal{T}_i$, the control input deviation is computed as $\xi(k) = K_i \eta(k)$ for all $k \in \mathcal{T}_i$, except during a dedicated data collection window.

To initialize this process, we require a controller for the first segment, $\mathcal{T}_0$; this controller can be synthesized using the known dynamics of the digital twin (or data collected from it offline), for instance, by designing a robust controller for the twin's linearization along the nominal trajectory. We then make the following assumption for initializing learning-based controllers.
\begin{assumption}[Initial Stabilizing Controller]\label{ass:initial_controller}
There exists an initial feedback gain $K_0$, designed based on the digital twin model, that is robustly stabilizing for the physical plant's deviation dynamics \eqref{eq:eta_ltv} over the initial time segment $\mathcal{T}_0$. Specifically, for any initial deviation $\eta(0)$ within a predefined set, the control law $\xi(k)= K_0 \eta(k)$ ensures that the state and input constraints $x(k)\in \mathcal{X}$ and $u(k)\in \mathcal{U}$ are satisfied for all $k \in \mathcal{T}_0$.
\end{assumption}

With initial safety guaranteed, the iterative process of controller refinement can begin. The first step involves gathering data from the plant to inform the next control policy.
At the end of each segment, we allocate a data collection window $\mathcal{T}_i^D$ to acquire informative data for updating the controller for the subsequent segment. Specifically,
\begin{equation}\label{eq:data_window}
    \mathcal{T}_i^D \coloneqq \{ k_i^D, k_i^D+1, \dots, k_{i+1}-1 \},
\end{equation}
where $k_i^D \coloneqq k_{i+1} - L$, and $L$ denotes the length of the data window. During $\mathcal{T}_i^D$, we introduce a bounded excitation signal $\epsilon(k)$ to ensure persistence of excitation. Thus, the control input deviation becomes,
\begin{equation}
    \xi(k) = K_i \eta(k) + \epsilon(k), \quad k \in \mathcal{T}_i^D.
\end{equation}
The excitation signal $\epsilon(k)$ is generated randomly and satisfies $\|\epsilon(k)\| \leq \bar{\epsilon}$ for all $k \in \mathcal{T}_i^D$, where $\bar{\epsilon} > 0$ is a predefined bound.

The data collected during $\mathcal{T}_i^D$ is used to construct the following matrices:
\begin{subequations}\label{eq:data_matrices_deviation}
\begin{align}
H_{i} & \coloneqq \begin{bmatrix}
    \eta(k_i^D) & \eta(k_i^D+1) & \cdots & \eta(k_{i+1}-1)
\end{bmatrix}, \label{eq:H_i} \\
H_{i}^+ & \coloneqq \begin{bmatrix}
    \eta(k_i^D+1) & \eta(k_i^D+2) & \cdots & \eta(k_{i+1})
\end{bmatrix}, \label{eq:H_i_plus} \\
\Xi_{i} & \coloneqq \begin{bmatrix}
    \xi(k_i^D) & \xi(k_i^D+1) & \cdots & \xi(k_{i+1}-1)
\end{bmatrix}. \label{eq:Xi_i}
\end{align}
\end{subequations}
Figure~\ref{fig:timeline-synthesis} summarizes how these matrices are gathered and then used to synthesize the next segment’s controller.

\begin{figure}[t]
  \centering
\includegraphics[width=0.9\columnwidth]{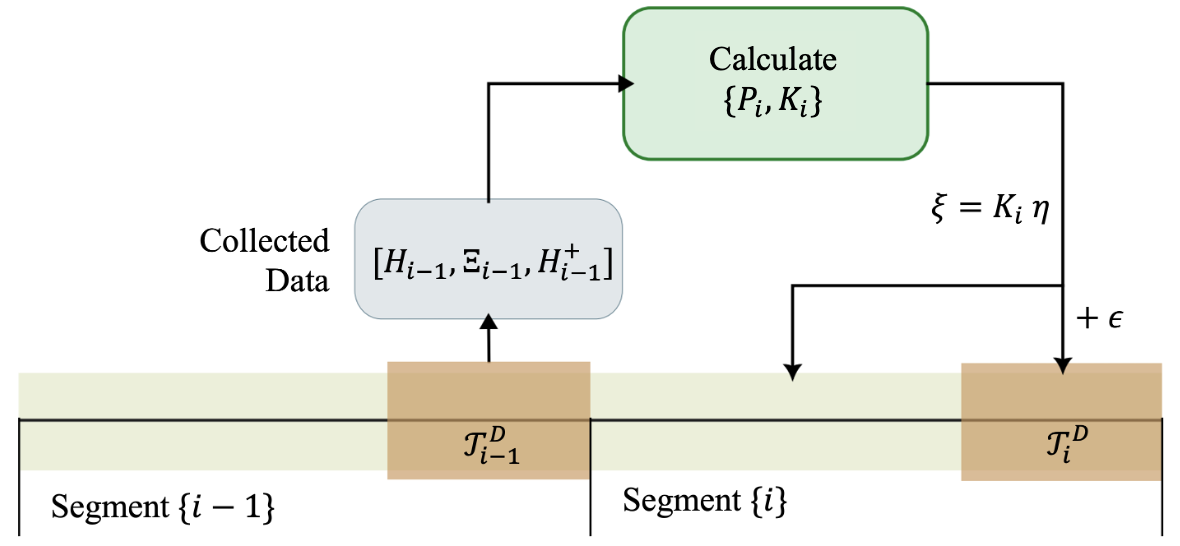}

  \caption{Segmented data collection and synthesis. During $\mathcal{T}^{D}_{i-1}$ we collect $(H_{i-1},\,\Xi_{i-1},\,H^{+}_{i-1})$. At $k_i$, these are passed to a convex synthesis step that returns $(P_i, K_i)$ for segment $\mathcal{T}_i$.}
  \label{fig:timeline-synthesis}
\end{figure}

To ensure the data are sufficiently informative for controller synthesis, we require that the persistence of excitation condition holds \cite{DePersis2019}:
\begin{equation}\label{eq:rank-condition}
    \operatorname{rank} \left( \begin{bmatrix} H_i \\ \Xi_i \end{bmatrix} \right) = n + m.
\end{equation}
A sufficient choice is to fix $L \ge n+m$, which makes \eqref{eq:rank-condition} achievable under bounded excitation.

\subsection{Data-Driven System Representation}

This subsection characterizes all systems consistent with the data collected during segment $i$, i.e., the matrices $H_i$, $H_i^+$, $\Xi_i$ defined in~\eqref{eq:data_matrices_deviation}. The goal is to construct uncertainty sets for the linearized dynamics that will be used for robust controller synthesis.

Let $A_i \coloneqq A(k_i)$ and $B_i \coloneqq B(k_i)$ denote the Jacobians of the true plant dynamics $f$ linearized at the nominal point $(\hat{x}_{\text{nom}}(k_i), \hat{u}_{\text{nom}}(k_i))$ at the start of segment $i$. For any $k \in \mathcal{T}_i^D$, the variation of the Jacobians throughout the segment is given by:
\[
\Delta A_{i,k} := A(k) - A_i, \quad \Delta B_{i,k} := B(k) - B_i.
\]

The deviation dynamics~\eqref{eq:eta_ltv} can thus be rewritten with a frozen model at $(A_i, B_i)$ and a compounded disturbance term:
\begin{equation}\label{eq:frozen_rep}
\eta(k+1) = A_i \eta(k) + B_i \xi(k) + w_i(k),
\end{equation}
where the disturbance $w_i(k)$ aggregates the effects of Jacobian variation and the original disturbance $d(k)$:
\[
w_i(k) \coloneqq \Delta A_{i,k} \, \eta(k) + \Delta B_{i,k} \, \xi(k) + d(k).
\]

\begin{lemma}[Bound on Total Disturbance]\label{lem:total_disturbance}
For any $k \in \mathcal{T}_i^D$, the total disturbance $w_i(k)$ in~\eqref{eq:frozen_rep} satisfies the following bound:
\[
\|w_i(k)\| \le C |k - k_i| \, \left\| \begin{bmatrix} \eta(k) \\ \xi(k) \end{bmatrix} \right\| + \gamma + L_r \left\| \begin{bmatrix} \eta(k) \\ \xi(k) \end{bmatrix} \right\|^2,
\]
where the constants are $C = L_J v$ (from Lemma~\ref{lem:AB_var}), $\gamma \ge 0$ (from Assumption~\ref{ass:mismatch_uniform}), and $L_r>0$ (from Lemma~\ref{lem:disturbance_bound}).
\end{lemma}

\begin{lemma}[Bound on Jacobian Variation]\label{lem:jacobian_variation}
For any $k \in \mathcal{T}_i \cup \mathcal{T}_{i+1}$, the variation $(\Delta A_{i,k}, \Delta B_{i,k})$ satisfies the following quadratic matrix inequality:
\begin{equation}\label{eq:bound-on-DeltaAB}
\big[\,\Delta A_{i,k}\ \ \Delta B_{i,k}\,\big]
\begin{bmatrix}
\Delta A_{i,k}^\top\\[2pt]
\Delta B_{i,k}^\top
\end{bmatrix}
\;\preceq\; C^2 \widetilde T_i^{\,2} I_n .
\end{equation}
where $\widetilde T_i \ge \max_{k \in \mathcal{T}_i \cup \mathcal{T}_{i+1}} |k - k_i|$ (e.g., $\widetilde T_i = 2T-1$), and $C = L_J v$ is the same constant as in Lemma~\ref{lem:total_disturbance}.
\end{lemma}

The following definition captures Jacobian variations consistent with the slow variation guaranteed by Lemma~\ref{lem:AB_var}.

\begin{definition}[Admissible Variation Set]\label{def:admissible_variation_set}
For segment $i$, the set of all admissible Jacobian variations is defined as:
\[
\Sigma_i^\Delta \coloneqq \left\{ (\Delta A, \Delta B) \in \mathbb{R}^{n \times n} \times \mathbb{R}^{n \times m} \, \middle| \, \eqref{eq:bound-on-DeltaAB} \text{ holds} \right\}.
\]
\end{definition}

Stacking the dynamics~\eqref{eq:frozen_rep} over the data collection window $\mathcal{T}_i^D$ yields the fundamental data equation:
\begin{equation}\label{eq:data_behavior}
H_i^{+} = A_i H_i + B_i \Xi_i + W_i,
\end{equation}
where $W_i := \begin{bmatrix} w_i(k_i^D) & \cdots & w_i(k_{i+1}-1) \end{bmatrix}$ is the stacked disturbance matrix. The pointwise bounds from Lemma~\ref{lem:total_disturbance} can be aggregated over the data window to form a quadratic constraint on $W_i$. This leads to the definition of the \emph{data-consistent system set}:

\begin{definition}[Data-Consistent System Set]\label{def:data_consistent_set}
The set of all matrices $(A_i, B_i)$ consistent with the data $(H_i, H_i^+, \Xi_i)$ and the disturbance bounds is given by:
\[
\Sigma_i \coloneqq \left\{ (A_i, B_i) \,\middle|\, \exists W_i \text{ such that } \eqref{eq:data_behavior} \text{ and } \eqref{eq:boundW} \text{ hold} \right\},
\]
\end{definition}
where the quadratic constraint
\begin{equation}\label{eq:boundW}
W_i W_i^\top \;\preceq\; \beta_i I.
\end{equation}
encodes the aggregated disturbance bound. The constant $\beta_i$ is an upper bound on the total energy of the disturbance sequence:
\begingroup
\setlength{\abovedisplayskip}{4pt}%
\setlength{\belowdisplayskip}{4pt}%
\begin{equation}\label{eq:beta_def}
\beta_i \coloneqq \sum_{k \in \mathcal{T}_i^D} \left( C\,|k - k_i| \, \left\| \begin{bmatrix} \eta(k) \\ \xi(k) \end{bmatrix} \right\| + \gamma + L_r \, \left\| \begin{bmatrix} \eta(k) \\ \xi(k) \end{bmatrix} \right\|^2 \right)^2.
\end{equation}
\endgroup

\begin{remark}\label{rem:beta_conservative}
The bound $\beta_i$ in~\eqref{eq:beta_def} is conservative,  as it aggregates worst-case point-wise bounds. This conservatism is a known trade-off in robust control that ensures safety guarantees. The online, data-driven nature of our algorithm mitigates this by continually refining the uncertainty set $\Sigma_i$ as new data arrives.
\end{remark}

\section{Data-Driven Funnel Synthesis}
\label{sec:data_driven_funnel}

In this section, at the start of each segment $i$, we design a funnel $\mathcal{F}_i$ (Def.~\ref{def:funnel}) together with a static feedback gain $K_i$ that keeps the deviation trajectories $(\eta,\xi)$ safe. The uncertainties are captured by the data-driven sets $\Sigma_{i-1}$ and $\Sigma_{i-1}^\Delta$ identified from the previous segment. 

Among all invariant-and-feasible funnels, we target the smallest one (in a volume sense) to tighten safety margins. To make the design tractable, we specialize $\mathcal{F}_i$ to ellipsoidal (Lyapunov) cross-sections of the state–input slices with a linear state-feedback— called a \emph{quadratic funnel}—which yields convex LMI conditions via the matrix S-lemma; see Fig.~\ref{fig:funnels}.

\subsection{Quadratic Funnel Framework}
\label{subsec:funnel_framework}
A quadratic funnel for segment $i$ is parameterized by a positive definite matrix $P_i \succ 0$ and a feedback gain matrix $K_i$. The $1$-level set of the Lyapunov function $V_i(\eta) = \eta^\top P_i \eta$ defines an ellipsoid of allowable state deviations:
\begin{equation}
\mathcal{E}(P_i) = \left\{ \eta \in \R^n \mid \eta^\top P_i \eta \le 1 \right\}.
\end{equation}

Under the linear control law $\xi = K_i \eta$, the induced input set is the ellipsoid
\[
\mathcal{E}_u(R_i) \;=\; \big\{ \xi \in \R^m \ \big|\ \xi^\top R_i^{-1} \xi \le 1 \big\},
\]
where $R_i = K_i P_i^{-1} K_i^\top$.

\begin{definition}[Quadratic Funnel]
A quadratic funnel for segment $i$ is the set in state–input space
\[
\mathcal{F}_i \;=\; \mathcal{E}(P_i) \times \mathcal{E}_u(R_i),
\]
parameterized by $P_i \succ 0$ and $K_i \in \R^{m \times n}$. The funnel is valid if:
\begin{itemize}
    \item \textbf{Invariance:} If $\eta(k) \in \mathcal{E}(P_i)$ at some $k \in \mathcal{T}_i$, then $\eta(k{+}1) \in \mathcal{E}(P_i)$ under all admissible uncertainties.
    \item \textbf{Feasibility:} The funnel lies within the state/input constraints:
    \begin{align*}
        \mathcal{E}(P_i) &\subseteq \left\{ \eta \in \R^n \mid \hat{x}_{\text{nom}}(k) + \eta \in \mathcal{X} \right\}, \\
        \mathcal{E}_u(R_i) &\subseteq \left\{ \xi \in \R^m \mid \hat{u}_{\text{nom}}(k) + \xi \in \mathcal{U} \right\}.
    \end{align*}
\end{itemize}
\end{definition}

\begin{figure}[t]
  \centering
  \def\svgwidth{\columnwidth}
  \begingroup\footnotesize
  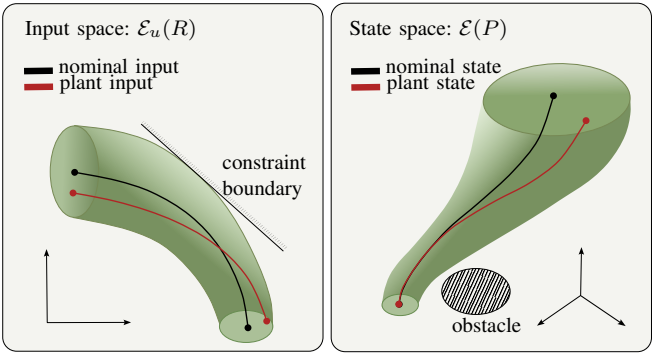
  \endgroup
  \caption{Schematic of per-segment quadratic funnels in input (left) and state (right) spaces. At time \(k\), the 1-level set \(\{\eta:\eta^\top P_i \eta \le 1\}\) is centered at \(\hat x_{\text{nom}}(k)\). Under \(\xi=K_i\eta\), any trajectory starting inside remains inside; the induced input ellipsoid \(\mathcal{E}_u(R_i)\) respects input limits.}
  \label{fig:funnels}
\end{figure}

To ensure feasibility, we assume that for each time $k$ there exist matrices
$P_{\min}(k)\!\succ\!0$ and $R_{\max}(k)\!\succeq\!0$ such that the deviation
ellipsoids
\begin{align*}
    \mathcal{E}\!\big(P_{\min}(k)\big)&\subseteq
\{\eta:\ \hat x_{\mathrm{nom}}(k)+\eta\in\mathcal X\},\\
\mathcal{E}_u\!\big(R_{\max}(k)\big)&\subseteq
\{\xi:\ \hat u_{\mathrm{nom}}(k)+\xi\in\mathcal U\}.
\end{align*}
For synthesis over a segment $\mathcal{T}_i$, we use single envelopes
$P_{\min,i}\!\succeq\!P_{\min}(k)$ and $R_{\max,i}\!\preceq\!R_{\max}(k)$ for all
$k\in\mathcal{T}_i$, which imply the per–segment feasibility conditions
$\mathcal{E}(P_i)\subseteq\mathcal{E}(P_{\min,i})$ and
$\mathcal{E}_u(R_i)\subseteq\mathcal{E}_u(R_{\max,i})$.
\begin{remark}[Computing $P_{\min}(k)$ and $R_{\max}(k)$]\label{rem:mvie}
We linearize each state constraint $h^x_j(x)\le 0$ at $\hat x_{\mathrm{nom}}(k)$ to obtain a local polytope in deviation coordinates:
$a^x_j(k)^\top \eta \le b^x_j(k)$ with
$a^x_j(k)=\nabla h^x_j(\hat x_{\mathrm{nom}}(k))$ and $b^x_j(k)=-h^x_j(\hat x_{\mathrm{nom}}(k))$.
Compute the maximum–volume inscribed ellipsoid via
\begin{align*}
P_{\min}(k)^{-1/2}
  &= \arg\max_{Z \succ 0}\ \log\det Z \\
\text{s.t.}\quad
  & \|Z\,a^x_j(k)\|_2 \le b^x_j(k), \quad \forall j,\\
  & 0 \preceq Z \preceq x_{\max} I.
\end{align*}
The input envelope $R_{\max}(k)$ can be computed in precisely the same manner in the $\xi$–coordinates using linearizations of the input constraints at $\hat u_{\mathrm{nom}}(k)$. See \cite{Boyd2004,Reynolds2020}.
\end{remark}

\begin{lemma}[State/input feasibility LMIs]\label{lem:feasibility_lmis}
Let $L_i:=K_iP_i$ with $P_i\succ0$. If
\[
\begin{bmatrix} R_{\max,i} & L_i\\ L_i^\top & P_i \end{bmatrix} \succeq 0
\quad\text{and}\quad
P_i \succeq P_{\min,i},
\]
then the induced input ellipsoid satisfies $\mathcal{E}_u(R_i)\subseteq\mathcal{E}_u(R_{\max,i})$ with $R_i=K_iP_i^{-1}K_i^\top$, and the state ellipsoid satisfies $\mathcal{E}(P_i)\subseteq\mathcal{E}(P_{\min,i})$.
\end{lemma}

For invariance, we require the Lyapunov function to decrease along closed-loop trajectories. With $\xi = K_i \eta$,
\[
\eta(k{+}1) = \big[A(k) + B(k) K_i\big] \eta(k) + d(k).
\]
A sufficient condition is the existence of $\alpha \in (0,1)$ such that
\begin{equation}\label{eq:lyapunov_decrease}
V_i(\eta(k{+}1)) \le \alpha\, V_i(\eta(k))
\end{equation}
for all $\eta(k) \in \mathcal{E}(P_i)$, all admissible $d(k)$, and all $(A(k),B(k))$ in the uncertainty description.

\begin{problem}[Data-Driven Robust Funnel Synthesis]\label{prob:funnel_synthesis}
For each segment $i \ge 1$, using the data-driven uncertainty sets $\Sigma_{i-1}$ and $\Sigma_{i-1}^\Delta$ from the previous segment, find $P_i \succ 0$ and $K_i$ such that:
\begin{enumerate}
    \item[(a)] \textbf{Robust Invariance:} \eqref{eq:lyapunov_decrease} holds for all $(A(k),B(k)) \in \Sigma_{i-1} \oplus \Sigma_{i-1}^\Delta$, all $\eta(k)\in\mathcal{E}(P_i)$, and all admissible disturbances.
    \item[(b)] \textbf{Feasibility:} Lemma~\ref{lem:feasibility_lmis}  holds.
\end{enumerate}
\end{problem}

In the following, we translate this problem into a tractable semi-definite program using LMIs and the Matrix S-Lemma.

\subsection{LMI-Based Synthesis}
\label{subsec:lmi_synthesis}

This subsection presents a tractable solution to Problem~\ref{prob:funnel_synthesis} using the informativity framework.
We first formalize when the available data are rich enough to certify a robust, feasible funnel from trajectories alone. This leads to an informativity notion tailored to our data-driven, per-segment synthesis.
\begin{definition}[Informativity for Quadratic Funnel Synthesis]
The data $(H_{i-1}, H^+_{i-1}, \Xi_{i-1})$ are informative if there exist $P_i \succ 0$ and $K_i$ such that the stability and feasibility conditions hold for all $(A,B) \in \Sigma_{i-1} \oplus \Sigma_{i-1}^\Delta$ (cf.\ \cite{vanWaarde2020TAC}).
\end{definition}

Theorem~\ref{thm:lmi_stability} provides a verifiable LMI certificate of informativity: if the LMI is feasible, the data are informative and one directly recovers $(P_i,K_i)$.

\begin{theorem}\label{thm:lmi_stability}
The robust stability condition \eqref{eq:lyapunov_decrease} holds for all $(A,B)\in\Sigma_{i-1}$ and $(\Delta A,\Delta B)\in\Sigma_{i-1}^\Delta$ if there exist $\lambda_1,\lambda_2\ge0$, $\nu>0$ as a slack parameter, and matrices $P_i\in \mathbb{S}^n_{++}$, $L_i \in \mathbb{R}^{m\times n}$ such that
\begin{equation}\label{eq:lmi_stability}
S(P_i, L_i, \nu)-\lambda_1 \widetilde{N}_1-\lambda_2 \widetilde{N}_2 \succ 0,
\end{equation}
where (with $L_i := K_i P_i$) and
\[
S(P_i,L_i,\nu)=
\begin{bmatrix}
\alpha P_i-\nu I & 0 & 0 & 0 & 0 & 0\\
0 & -P_i & -L_i^\top & -P_i & -L_i^\top & 0\\
0 & -L_i & 0 & -L_i & 0 & L_i\\
0 & -P_i & -L_i^\top & -P_i & -L_i^\top & 0\\
0 & -L_i & 0 & -L_i & 0 & L_i\\
0 & 0 & L_i^\top & 0 & L_i^\top & P_i
\end{bmatrix}.
\]

\begin{align*}
&\widetilde N_1\;=\;
\begin{bmatrix}
N_1 & 0\\
0   & 0
\end{bmatrix},\\
&N_1 \;=\;
\begin{bmatrix}
I_n & {H_{i-1}^+} \\  0& -H_{i-1} \\ 0 & -\Xi_{i-1} \\ 0 & 0 \\ 0& 0
\end{bmatrix}
\begin{bmatrix}
    \beta_{i-1}I & 0 \\ 0 & -I
\end{bmatrix}
\begin{bmatrix}
I_n & {H_{i-1}^+} \\  0& -H_{i-1} \\ 0 & -\Xi_{i-1} \\ 0 & 0 \\ 0& 0
\end{bmatrix}^\top
\end{align*}

\begin{align*}
    &\widetilde{N}_2 = 
    \begin{bmatrix}
N_2 & 0\\
0   & 0
\end{bmatrix},\\
    & N_2 = \begin{bmatrix}
I_n & 0 & 0\\
0 & 0 &0 \\
0 & 0 & 0\\
0 & I_n & 0\\
0 & 0 & I_m
\end{bmatrix} \begin{bmatrix}
C^2 \widetilde T_i^{\,2} I & 0 & 0\\
0 & -I & 0\\
0 & 0 & -I
\end{bmatrix}\begin{bmatrix}
I_n & 0 & 0\\
0 & 0 &0 \\
0 & 0 & 0\\
0 & I_n & 0\\
0 & 0 & I_m
\end{bmatrix}^\top
\end{align*}
\end{theorem}

\begin{proof}
Fix segment $i$. Let $V_i(\eta)=\eta^\top P_i\,\eta$ with $P_i=P_i^\top\succ0$.
Condition \eqref{eq:lyapunov_decrease} with $0<\alpha<1$ is equivalent to the matrix inequality
\begin{equation}\label{eq:design-pre}
\alpha P_i - A_{\mathrm{cl},i} P_i A_{\mathrm{cl},i}^\top \;\succeq\;0, 
\end{equation}
where $A_{\mathrm{cl},i}=A_{i}+B_{i}K_i+\Delta A_{i,k}+\Delta B_{i,k}K_i$. Expanding the quadratic term produces a quadratic matrix inequality (QMI) in the unknowns $(A_{i},B_{i},\Delta A_{i,k},\Delta B_{i,k})$ and in $K_i$. Define the lifted stack
\[
\Phi_i \;=\; \begin{bmatrix} I\\ A_{i}^\top\\ B_{i}^\top\\ \Delta A_{i,k}^\top\\ \Delta B_{i,k}^\top \end{bmatrix}.
\]
Then \eqref{eq:design-pre} is equivalent to $\Phi_i^\top M_i\,\Phi_i \;\succeq\; 0,$
with
\[
M_i=
\begin{bmatrix}
\alpha P_i & 0 & 0 & 0 & 0\\
0 & -P_i & -P_iK_i^\top & -P_i & -P_iK_i^\top\\
0 & -K_iP_i & -K_iP_iK_i^\top & -K_iP_i & -K_iP_iK_i^\top\\
0 & -P_i & -P_iK_i^\top & -P_i & -P_iK_i^\top\\
0 & -K_iP_i & -K_iP_iK_i^\top & -K_iP_i & -K_iP_iK_i^\top
\end{bmatrix}.
\]

The data/noise relation \eqref{eq:boundW} with
$W_i = H_i^{+} - A_{i} H_i - B_{i}\Xi_i$
can be written as a QMI in $\Phi_i$: $\Phi_i^\top\,N_1\,\Phi_i \;\succeq\; 0$ where $N_1 \coloneqq S_W\,M_W\,S_W^\top$, $M_W=\diag(\beta_{i-1} I,\,-I)$, and
\begin{align*}
    &S_W^\top =
\begin{bmatrix}
I & 0 & 0 & 0 & 0\\
H_{i-1}^{+\top} & -H_{i-1}^\top & -\Xi_{i-1}^\top & 0 & 0
\end{bmatrix}.\\
\end{align*}
Likewise, the uncertainty bound \eqref{eq:bound-on-DeltaAB} is $\Phi_i^\top\,N_2\,\Phi_i \;\succeq\; 0$ where $N_2 \coloneqq S_\Delta^\top\,M_\Delta\,S_\Delta,$
and $M_\Delta=\diag\!\big(C^2\widetilde T_i^{\,2} I_n,\,-I_n,\,-I_m\big)$ and
\begin{align*}
    &S_\Delta =
\begin{bmatrix}
I & 0 & 0 & 0 & 0\\
0 & 0 & 0 & I & 0\\
0 & 0 & 0 & 0 & I
\end{bmatrix}.
\end{align*}

We require this to hold for all $(A_{i},B_{i},\Delta A_{i,k},\Delta B_{i,k})$ that satisfy the two premise QMIs above.
Under a Slater-type feasibility, the matrix S-lemma yields scalars
$\lambda_1,\ \lambda_2,\ \nu>0$ such that
\begin{equation}\label{eq:S-lemma-5x5}
M\;-\;\lambda_1 N_1\;-\;\lambda_2 N_2\;\succeq\;\diag(\nu I_n,\,0,\,0,\,0,\,0).
\end{equation}

Introduce the standard change of variables $L_i\coloneqq K_iP_i$ and add an auxiliary $P$-block so that all occurrences of $-KPK^\top$ arise from a Schur complement. Define $S(P_i, L_i,\nu)$ as in Theorem~\ref{thm:lmi_stability}.
Pad the premise QMIs with a zero row/column matching the auxiliary blocks $\widetilde N_1$ and $\widetilde N_2$.
Then \eqref{eq:S-lemma-5x5} is equivalent to the linear matrix inequality (see \cite{vanWaarde2020MatrixSLemma}) \eqref{eq:lmi_stability} and the controller can be recovered as $K_i=L_i P_i^{-1}$.
\end{proof}

Theorem~\ref{thm:lmi_stability} provides an LMI certificate of informativity and robust one–step decay, while Lemma~\ref{lem:feasibility_lmis} encodes state/input set containment as LMIs. These yield the following convex program for computing $(P_i,K_i)$ at segment $i$; its constraints enforce robustness and feasibility, and the $\log\det(P_i)$ objective minimizes the volume of the state slice of the funnel.

\begin{equation}\label{eq:main_sdp}
\begin{aligned}
& \underset{P_i, L_i, \lambda_1, \lambda_2,\nu}{\text{maximize}} && \log\det(P_i) \\
& \text{subject to} \quad && P_i \succ 0,\ \lambda_1 \ge 0,\ \lambda_2 \ge 0, \\
&&& S(P_i,L_i,\nu) - \lambda_1 \widetilde N_1 - \lambda_2 \widetilde N_2 \succeq 0, \\
&&& P_i \succeq P_{\min,i}, \\
&&& \begin{bmatrix} R_{\max,i} & L_i \\ L_i^\top & P_i \end{bmatrix} \succeq 0 . 
\end{aligned}
\end{equation}

\subsection{Stability Analysis}
\label{subsec:stability_analysis}

This subsection establishes the closed-loop stability guarantees for the proposed segmented control strategy. The key result shows that, under the data-driven funnel synthesis, the deviation dynamics exhibit practical exponential stability within the synthesized funnels.

\begin{theorem}[Practical Exponential Stability within Funnels]
\label{thm:pges_main}
Let the deviation dynamics be given by \eqref{eq:eta_ltv} under the segmented policy on $\mathcal{T}_i=\{k_i,\dots,k_i+T-1\}$. 
Assume for every $i\ge0$:
\begin{enumerate}
\item[(i)] The LMI in Theorem~\ref{thm:lmi_stability} is feasible, yielding $(P_i\succ0,K_i)$ and ensuring the robust one-step decay
$V_i(\eta(k{+}1))\le \alpha\,V_i(\eta(k))$ for all $\eta(k)\in\mathcal E(P_i)$, all admissible $(A(k),B(k))$, and all admissible disturbances $d(k)$.
\item[(ii)] Uniform bounds $P_{\min}\preceq P_i\preceq P_{\max}$ with $p_{\min}:=\lambda_{\min}(P_{\min})$ and $p_{\max}:=\lambda_{\max}(P_{\max})$.
\item[(iii)] Cross-segment growth $P_{i+1}\preceq \mu\,P_i$ for some $\mu\ge1$.
\item[(iv)] During data windows $\mathcal{T}_i^D$, the excitation satisfies $\|\epsilon(k)\|\le\bar\epsilon$, and $\|B(k)\|\le\bar B$ holds for all $k$.
\end{enumerate}
If the dwell time satisfies $T>-\ln\mu/\ln\alpha$, then with $\hat\alpha:=\alpha\,\mu^{1/T}\in(\alpha,1)$ we have, for all $k\in\mathbb N$ with $\eta(k)\in\mathcal D:=\bigcup_i\mathcal E(P_i)$,
\begin{equation}
\label{eq:pges_bound_final_norm}
\|\eta(k)\|\ \le\
\sqrt{\tfrac{p_{\max}}{p_{\min}}}\;\hat{\alpha}^{k/2}\,\|\eta(0)\|
\ +\
\sqrt{\tfrac{p_{\max}}{p_{\min}}}\;
\frac{\big(\hat{\alpha}/\alpha\big)^{T/2}}{1-\sqrt{\hat{\alpha}}}\;\bar B\,\bar\epsilon.
\end{equation}
\end{theorem}

\begin{proof}
Define, for $k\in\mathcal{T}_i$,
\[
q(k):=\Big(\tfrac{\hat\alpha}{\alpha}\Big)^{\frac{k-k_i}{2}}\big\|P_i^{1/2}\eta(k)\big\|.
\]
For $k\in\mathcal{T}_i\setminus\mathcal{T}_i^D$, item (i) gives
\[
\big\|P_i^{1/2}\!\big((A(k){+}B(k)K_i)\eta(k)+d(k)\big)\big\|
\le \sqrt{\alpha}\,\big\|P_i^{1/2}\eta(k)\big\|,
\]
hence $q(k{+}1)\le \sqrt{\hat\alpha}\,q(k)$. 
For $k\in\mathcal{T}_i^D$ we apply $\xi=K_i\eta+\epsilon$, yielding
\[
\big\|P_i^{1/2}\eta(k{+}1)\big\|
\le \sqrt{\alpha}\,\big\|P_i^{1/2}\eta(k)\big\|
     + \big\|P_i^{1/2}B(k)\epsilon(k)\big\|.
\]
Using $\|B(k)\|\le\bar B$, $\|\epsilon(k)\|\le\bar\epsilon$, and $\|P_i^{1/2}\|\le \sqrt{p_{\max}}$ gives
\begin{align*}
    q(k{+}1)&\le \sqrt{\hat\alpha}\,q(k)
           + \Big(\tfrac{\hat\alpha}{\alpha}\Big)^{\!\frac{k+1-k_i}{2}}\sqrt{p_{\max}}\,\bar B\,\bar\epsilon
           \\ &\le\ 
           \sqrt{\hat\alpha}\,q(k)
           + \Big(\tfrac{\hat\alpha}{\alpha}\Big)^{\!\frac{T}{2}}\sqrt{p_{\max}}\,\bar B\,\bar\epsilon.
\end{align*}
At $k=k_{i+1}$, item (iii) gives $P_{i+1}\preceq \mu P_i$, hence
$\|P_{i+1}^{1/2}\eta(k_{i+1})\|\le \mu^{1/2}\|P_i^{1/2}\eta(k_{i+1})\|
= \big(\tfrac{\hat\alpha}{\alpha}\big)^{T/2}\|P_i^{1/2}\eta(k_{i+1})\|$, so the definition of $q(\cdot)$ is consistent across the boundary. Unrolling the scalar recursion yields
\[
\big\|P_i^{1/2}\eta(k)\big\|
\le \hat\alpha^{k/2}\big\|P_0^{1/2}\eta(0)\big\|
   + \frac{(\hat\alpha/\alpha)^{T/2}}{1-\sqrt{\hat\alpha}}\;\sqrt{p_{\max}}\,\bar B\,\bar\epsilon.
\]
Finally, for any $i$, $\|\eta\|\le \frac{1}{\sqrt{\lambda_{\min}(P_i)}}\,\|P_i^{1/2}\eta\|\le \frac{1}{\sqrt{p_{\min}}}\|P_i^{1/2}\eta\|$ and
$\|P_0^{1/2}\eta(0)\|\le \sqrt{\lambda_{\max}(P_0)}\,\|\eta(0)\|\le \sqrt{p_{\max}}\,\|\eta(0)\|$, which yields \eqref{eq:pges_bound_final_norm}.
\end{proof}

\begin{remark}[Domain of Attraction]
The result holds on the regional domain $\mathcal{D}=\bigcup_i\mathcal{E}(P_i)$. Extending beyond $\mathcal{D}$ would require global assumptions or funnel coverage. The bound in \eqref{eq:pges_bound_final_norm} is conservative with respect to the excitation schedule, as it upper-bounds by the worst case over each segment.
\end{remark}

The stability analysis confirms that the proposed control framework provides robust stability guarantees within the synthesized funnels. The segmented approach enables online controller refinement while preserving stability through carefully managed transitions, with explicit bounds characterizing the trade-offs between adaptation frequency, excitation magnitude, and convergence rate. The corresponding online procedure is summarized in Algorithm~\ref{alg:online_control}.

\begin{algorithm}[t]
\caption{Online Data-Driven Funnel Control}
\label{alg:online_control}
\begin{algorithmic}[1]
\Require Nominal trajectory; segment length $T$; data window $L$; decay $\alpha\in(0,1)$; per-time $P_{\min}(k),R_{\max}(k)$; bound $\bar\epsilon$; $(\gamma,L_J,L_r,\widetilde T)$; initial $(P_0,K_0)$
\State $i\gets0$, $k_0\gets0$; define $\mathcal{T}_i=\{k_i,\dots,k_i{+}T{-}1\}$; controller $(P_0,K_0)$ is \emph{given}
\While{$k_i<N$}
  \State \textbf{Execute on $\mathcal{T}_i$:} set $\mathcal{T}_i^D=\{k_{i+1}{-}L,\dots,k_{i+1}{-}1\}$
  \For{$k\in\mathcal{T}_i$}
    \State apply $\xi(k)=K_i\eta(k)$; \textbf{if} $k\in\mathcal{T}_i^D$ \textbf{then} set $\xi(k)= K_i\eta(k)+\epsilon(k)$ with $\|\epsilon(k)\|\le\bar\epsilon$
    \State log $(\eta(k),\xi(k),\eta(k{+}1))$
  \EndFor
  \State \textbf{Assemble data at $k_{i+1}$:} form $H_i,H_i^+,\Xi_i$; choose $P_{\min,i}\succeq P_{\min}(k)$, $R_{\max,i}\preceq R_{\max}(k)$ on $\mathcal{T}_i$; compute $\beta_i$; build $\Sigma_i$ and $\Sigma_i^\Delta$
  \State \textbf{Design for next segment $\mathcal{T}_{i+1}$ (before entering it):} solve \eqref{eq:main_sdp} (with $P_{i+1}\preceq \mu P_i$, to get $(P_{i+1},L_{i+1})$; set $K_{i+1}=L_{i+1}P_{i+1}^{-1}$
  \State $k_{i+1}\gets k_i{+}T$;\quad $i\gets i{+}1$;\quad update $\mathcal{T}_i$
\EndWhile
\end{algorithmic}
\end{algorithm}

\section{Case Study}\label{sec:case}
We consider a planar 2-DoF robot arm with standard rigid-body dynamics \cite{Spong2006}. 
Let $x=[q_1,q_2,\dot q_1,\dot q_2]^\top\!\in\mathbb{R}^4$ and $u=[\tau_1,\tau_2]^\top\!\in\mathbb{R}^2$, 
where $q_1$ and $q_2$ are the first (shoulder) and second (elbow) joint angles , 
$\dot q_1$ and $\dot q_2$ are their angular velocities, and 
$\tau_1$ and $\tau_2$ are the actuator torques applied at the two joints.
The continuous-time dynamics are
\[
\dot{x}=\begin{bmatrix}\dot q \\ M(q)^{-1}\big[\tau - C(q,\dot q)\dot q - G(q) - B\dot q\big]\end{bmatrix},
\]
where $M(q)\in\mathbb{R}^{2\times2}\!\succ0$ is the mass matrix, $C(q,\dot q)$ the Coriolis/centrifugal term, $G(q)$ gravity, and $B$ diagonal viscous friction:
\begin{align*}
M(q) &= \begin{bmatrix}
a + 2b\cos q_2 & d + b\cos q_2 \\
d + b\cos q_2 & d
\end{bmatrix},\\
C(q,\dot q) &= \begin{bmatrix}
-2b\sin q_2\,\dot q_2 & -b\sin q_2\,\dot q_2 \\
b\sin q_2\,\dot q_1 & 0
\end{bmatrix}, \quad B = \begin{bmatrix} b_1 & 0 \\ 0 & b_2 \end{bmatrix},\\
G(q) &= \begin{bmatrix}
(m_1l_{c1} + m_2l_1)g\cos q_1 + m_2l_{c2}g\cos(q_1+q_2) \\
m_2l_{c2}g\cos(q_1+q_2)
\end{bmatrix},
\end{align*}
with $a = I_1 + I_2 + m_1 l_{c1}^2 + m_2(l_1^2 + l_{c2}^2)$, $b = m_2 l_1 l_{c2}$, and $d = I_2 + m_2 l_{c2}^2$. Here, $m_i$ are link masses, $l_i$ link lengths, $l_{ci}$ center-of-mass distances, $I_i$ link inertias about the center-of-mass, and $b_i$ viscous friction coefficients. We use the standard parameter set $(m_i,l_i,l_{ci},b_i,I_i)$ with gravity $g=9.81$\,m/s$^2$. Plant/twin parameters (slightly mismatched) are in Table~\ref{tab:parameters}.

\begin{table}[ht]
\centering
\caption{Physical parameters of the 2-DoF arm (plant vs. twin).}
\label{tab:parameters}
\setlength{\tabcolsep}{4pt}
\renewcommand{\arraystretch}{1.05}
\begin{minipage}[t]{0.48\columnwidth}
\centering
\begin{tabular}{lcc}
\toprule
Parameter & Plant & Twin \\
\midrule
$m_1$ (kg)             & 1.00  & 0.95  \\
$l_1$ (m)              & 0.70  & 0.73  \\
$l_{c1}$ (m)           & 0.35  & 0.365 \\
$I_1$ (kg$\cdot$m$^2$) & 0.050 & 0.055 \\
$b_1$ (N$\cdot$m$\cdot$s/rad) & 0.020 & 0.018 \\
\bottomrule
\end{tabular}
\end{minipage}\hfill
\begin{minipage}[t]{0.48\columnwidth}
\centering
\begin{tabular}{lcc}
\toprule
Parameter & Plant & Twin \\
\midrule
$m_2$ (kg)             & 0.80  & 0.84  \\
$l_2$ (m)              & 0.60  & 0.58  \\
$l_{c2}$ (m)           & 0.30  & 0.29  \\
$I_2$ (kg$\cdot$m$^2$) & 0.040 & 0.038 \\
$b_2$ (N$\cdot$m$\cdot$s/rad) & 0.020 & 0.022 \\
\bottomrule
\end{tabular}
\end{minipage}
\end{table}

State/input boxes:
$q_1\!\in[-5,9],\ q_2\!\in[-8,8],\ \dot q_1\!\in[-8,8],\ \dot q_2\!\in[-7,7]$ and $\tau_1,\tau_2\!\in[-40,40]$.
We discretize with RK4 ($\Delta t=0.01$\,s) and simulate $N=600$ steps (6\,s). A twin-based LQR plans a nominal from $\hat x(0)=[0.28,-0.22,0,0]^\top$ to $x_{\mathrm{goal}}=[4.0,-1.0,0,0]^\top$ with $u_{\mathrm{goal}}=[-8.52,-2.37]^\top$ N$\cdot$m.
Numerical constants for synthesis are $\gamma=0.034$ (max one-step plant–twin mismatch along the nominal), $L_J=16.94$ (Jacobian Lipschitz estimate from finite-difference probing of the discrete-time map near the nominal trajectory), and $L_r=0.074$ (second-order remainder constant from finite-difference sampling). We use $C=L_J\,v=3.64$ in the Jacobian-variation bound. We use $\alpha=0.98$ as the Lyapunov decay rate.

We compute time-varying maximum–volume inscribed ellipsoids along the nominal (cf. \cite{Boyd2004}), yielding $P_{\min,i},R_{\max,i}$ used in Lemma~\ref{lem:feasibility_lmis}. At $k=0$:
$P_{\min}^{1/2}(0)=\diag(0.19,\,0.12,\,0.12,\,0.14)$ and
$R_{\max}^{1/2}(0)=\diag(28.56,\,37.62)$.

Segments have length $T=100$ (1\,s); data-window length $L=60$, with excitation bound $\bar\epsilon=0.15$; variation constants $C=L_J v=3.64$ and $\widetilde T=199$. We impose $P_{i+1}\preceq 1.02\,P_i$. The plant starts at $x(0)=[2.28,1.78,1.0,-1.0]^\top$.

The SDP is posed in \texttt{CVXPY} \cite{Diamond2016} and solved with \texttt{SCS}. Figure~\ref{fig:q1234} shows the plant (red) tracking the nominal (black) within the certified funnels (gray) and constraints (green). The baseline—directly applying the twin's nominal controller to the plant (dashed gray)—shows large excursions and repeated violations. The terminal deviation is $\|\eta(N)\|=4.725\times10^{-2}$.

\begin{figure}[t]
  \centering
  \includegraphics[width=1\columnwidth]{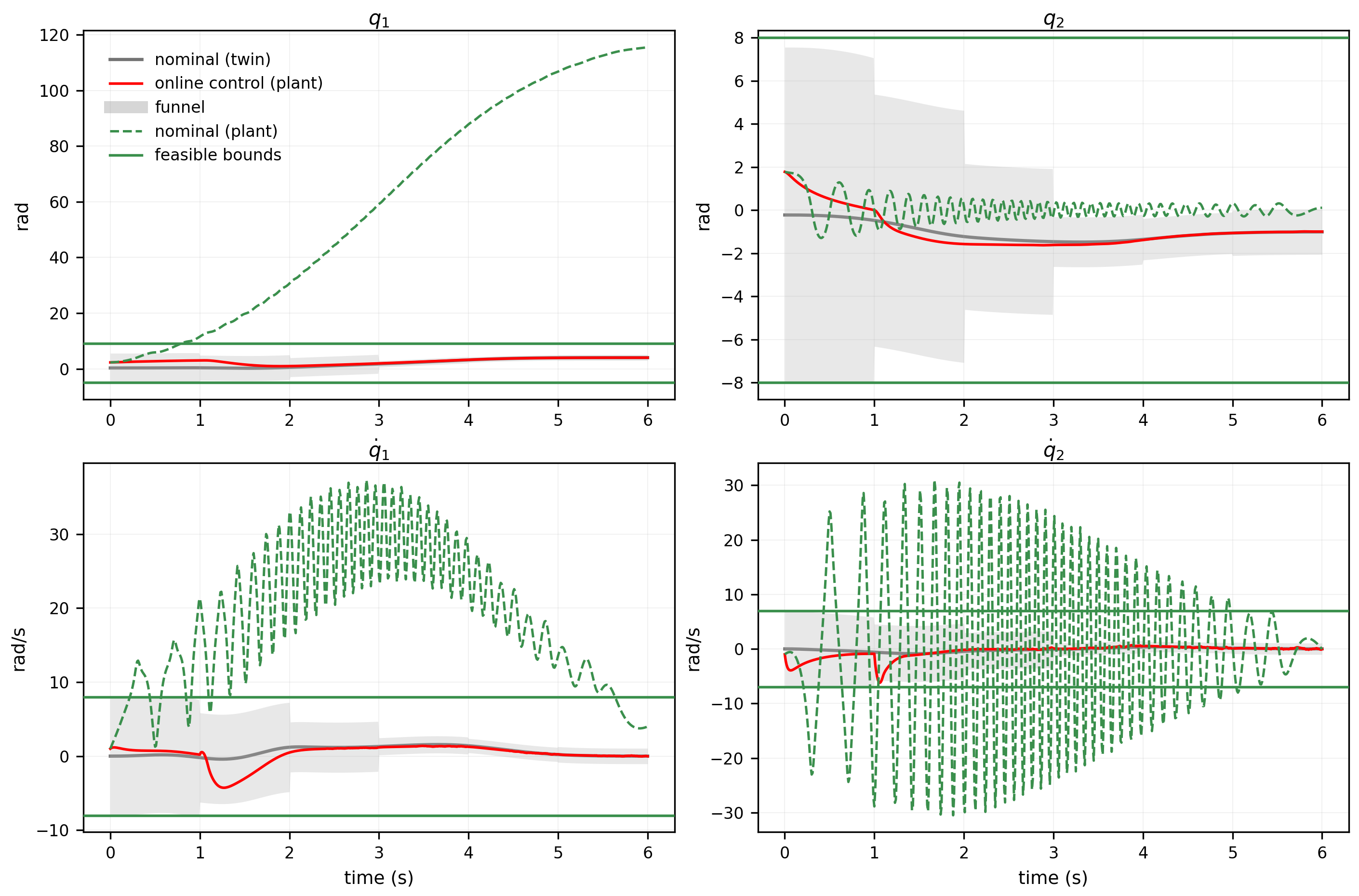}
  \caption{Closed-loop tracking on the 2-DoF arm. The proposed data-driven controller (red) tracks the nominal trajectory computed on the digital twin (black) while remaining within the certified state funnels (gray) and respecting state constraints (green). For comparison, applying the twin controller directly to the plant (dashed gray) leads to larger excursions and constraint violations.}
  \label{fig:q1234}
\end{figure}

\section{Conclusion}\label{sec:conclusion}
This paper proposed a data-driven control framework utilizing the nominal trajectory and control obtained from an imperfect digital twin of a nonlinear system for the purpose of its control with formal safety guarantees. In this direction, plant–twin deviation is first represented via locally LTV dynamics, from which finite trajectories yielded data-consistent uncertainty sets capturing linearization error and Jacobian variations. Subsequently, time-varying quadratic funnels were synthesized over these sets via data-parameterized LMIs, certifying robust positive invariance and constraint satisfaction. The proposed method was implemented online in segments, updating certificates and feedback from newly acquired data. 
Future work includes output-feedback, tighter online and model-free estimation of Lipschitz constants, and less conservative descriptions and parameterization of uncertainty.
\bibliographystyle{IEEEtran}
\bibliography{ref}
\end{document}

%% file: figures/funnnels.pdf_tex
\begingroup%
  \makeatletter%
  \providecommand\color[2][]{%
    \errmessage{(Inkscape) Color is used for the text in Inkscape, but the package 'color.sty' is not loaded}%
    \renewcommand\color[2][]{}%
  }%
  \providecommand\transparent[1]{%
    \errmessage{(Inkscape) Transparency is used (non-zero) for the text in Inkscape, but the package 'transparent.sty' is not loaded}%
    \renewcommand\transparent[1]{}%
  }%
  \providecommand\rotatebox[2]{#2}%
  \newcommand*\fsize{\dimexpr\f@size pt\relax}%
  \newcommand*\lineheight[1]{\fontsize{\fsize}{#1\fsize}\selectfont}%
  \ifx\svgwidth\undefined%
    \setlength{\unitlength}{232.27612833bp}%
    \ifx\svgscale\undefined%
      \relax%
    \else%
      \setlength{\unitlength}{\unitlength * \real{\svgscale}}%
    \fi%
  \else%
    \setlength{\unitlength}{\svgwidth}%
  \fi%
  \global\let\svgwidth\undefined%
  \global\let\svgscale\undefined%
  \makeatother%
  \begin{picture}(1,0.53818875)%
    \lineheight{1}%
    \setlength\tabcolsep{0pt}%
    \put(0,0){\includegraphics[width=\unitlength,page=1]{funnnels.pdf}}%
    \put(0.03508736,0.48915774){\color[rgb]{0,0,0}\makebox(0,0)[lt]{\lineheight{1.25}\smash{\begin{tabular}[t]{l}Input space: $\mathcal{E}_u(R)$\end{tabular}}}}%
    \put(0.53424803,0.48889129){\color[rgb]{0,0,0}\makebox(0,0)[lt]{\lineheight{1.25}\smash{\begin{tabular}[t]{l}State space: $\mathcal{E}(P)$\end{tabular}}}}%
    \put(0,0){\includegraphics[width=\unitlength,page=2]{funnnels.pdf}}%
    \put(0.08707795,0.42997433){\color[rgb]{0,0,0}\makebox(0,0)[lt]{\lineheight{1.25}\smash{\begin{tabular}[t]{l}nominal input\end{tabular}}}}%
    \put(0.69232037,0.03046529){\color[rgb]{0,0,0}\makebox(0,0)[lt]{\lineheight{1.25}\smash{\begin{tabular}[t]{l}obstacle\\\end{tabular}}}}%
    \put(0.33923353,0.28218621){\color[rgb]{0,0,0}\makebox(0,0)[lt]{\lineheight{1.25}\smash{\begin{tabular}[t]{l}constraint\\boundary\end{tabular}}}}%
    \put(0.0884705,0.40530797){\color[rgb]{0,0,0}\makebox(0,0)[lt]{\lineheight{1.25}\smash{\begin{tabular}[t]{l}plant input\end{tabular}}}}%
    \put(0,0){\includegraphics[width=\unitlength,page=3]{funnnels.pdf}}%
    \put(0.59056325,0.42948029){\color[rgb]{0,0,0}\makebox(0,0)[lt]{\lineheight{1.25}\smash{\begin{tabular}[t]{l}nominal state\end{tabular}}}}%
    \put(0.59195574,0.40481391){\color[rgb]{0,0,0}\makebox(0,0)[lt]{\lineheight{1.25}\smash{\begin{tabular}[t]{l}plant state\end{tabular}}}}%
  \end{picture}%
\endgroup%